\documentclass[12pt,reqno]{article}

\usepackage[usenames]{color}
\usepackage{amssymb}
\usepackage{amsmath}
\usepackage{amsthm}
\usepackage{amsfonts}
\usepackage{amscd}
\usepackage{graphicx}
\usepackage{xcolor}

\usepackage[colorlinks=true,
linkcolor=webgreen,
filecolor=webbrown,
citecolor=webgreen]{hyperref}

\definecolor{webgreen}{rgb}{0,.5,0}
\definecolor{webbrown}{rgb}{.6,0,0}

\usepackage{color}
\usepackage{fullpage}
\usepackage{float}

\usepackage{graphics}
\usepackage{latexsym}
\usepackage{epsf}
\usepackage{breakurl}
\usepackage{fullpage}

\DeclareMathOperator{\ord}{ord}

\begin{document}

\theoremstyle{plain}
\newtheorem{theorem}{Theorem}
\newtheorem{corollary}[theorem]{Corollary}
\newtheorem{lemma}[theorem]{Lemma}
\newtheorem{proposition}[theorem]{Proposition}

\theoremstyle{definition}
\newtheorem{definition}[theorem]{Definition}
\newtheorem{example}[theorem]{Example}
\newtheorem{conjecture}[theorem]{Conjecture}

\theoremstyle{remark}
\newtheorem{remark}[theorem]{Remark}

\title{Kleene Star of the Primes is not Regular in Any Base}

\author{Jason Yuen\\
School of Computer Science\\
University of Waterloo\\
Waterloo, ON  N2L 3G1\\
Canada\\
\href{mailto:j35yuen@uwaterloo.ca}{\tt j35yuen@uwaterloo.ca}}

\maketitle

\begin{abstract}
Let $P_b$ denote the primes expressed in base-$b$. In this note, we prove that $P_b^*$ is not regular. This strengthens a classical result that $P_b$ is not regular, due to Minsky and Papert in 1966.
\end{abstract}

\section{Introduction}

Let $P_b$ be the language of primes expressed in base $b \ge 2$, using the most significant digit first, and disallowing leading zeroes. Minsky and Papert \cite{Minsky&Papert:1966} proved in 1966 that $P_b$ is not regular. Furthermore, Shallit \cite{Shallit:1996} proved in 1996 that any finite automaton that correctly recognizes primes up to $N$ must have $\Omega(N^{1/43})$ states. Hartmanis and Shank \cite{Hartmanis&Shank:1968}, and, independently, Sch\"utzenberger \cite{Schutzenberger:1968} proved in 1968 that the primes are not context-free.

We show a new result that $P_b^*$ is not regular. Firstly, assume $P_b^*$ is regular. Using number theory, we find a prime of the form $K \times b^n + 1$ where $K$ is sufficiently large, and $k \times b^n + 1$ is composite for all $k < K$. Using the classical pumping lemma for regular languages, $P_b^*$ contains $k \times b^n + 1$ for some $k < K$. This leads to a contradiction.

\section{The main theorem}

\begin{theorem}[Dirichlet]
Let $a, d \ge 1$ be coprime. Then there are infinitely many primes of the form $a + kd$ where $k \ge 1$.
\end{theorem}

\begin{definition}
For each $b \ge 2$, define $f_b : \mathbb Z^+ \to \mathbb Z^+$ as follows. For each $n$, let $f_b(n)$ be the smallest $k$ such that $k \times b^n + 1$ is prime. By Dirichlet's Theorem with $(a, d) = (1, b^n)$, there are infinitely many primes of the form $k \times b^n + 1$ with $k \ge 1$, thus $f_b(n)$ is well-defined.
\end{definition}

\begin{proposition}\label{bigf}
Let $b \ge 2$ and $K \ge 1$. There exists $N \ge 2$ such that $f_b(N) > K$.
\end{proposition}
\begin{proof}
Let $N = (bK)! + 1$. Note that $N \ge 2!+1 > 2$. It remains to show $f_b(N) > K$.

Let $1 \le k \le K$. We show that $bk+1$ is a factor of $k \times b^N + 1$. Firstly, consider $\ord(b)$ in $\mathbb Z_{bk+1}$. Since $\gcd(b, bk+1) = 1$, we have $b \in \mathbb Z_{bk+1}^*$ and $\ord(b) \le bk \le bK$. Therefore, $b^{(bK)!} \equiv 1 \pmod{bk+1}$. Compute $k \times b^N + 1$ in $\mathbb Z_{bk+1}$:
$$
k \times b^N + 1
\equiv bk \times b^{(bK)!} + 1
\equiv bk \times 1 + 1
\equiv 0 \pmod{bk+1}
$$

Therefore, $bk+1 \mid k \times b^N + 1$. Firstly, $bk+1 \ge 2+1 > 1$. Secondly, $N \ge 2$ implies $bk+1 < k \times b^N + 1$. Therefore, $k \times b^N + 1$ is composite. It follows that $f_b(N) > K$.
\end{proof}

\begin{definition}
For $b \ge 2$ and $n \ge 0$, let $(n)_b$ be $n$ expressed in base $b$, with the most significant digit first, and with no leading zeroes.
\end{definition}

\begin{lemma}\label{noP*}
Let $b \ge 2$ and $n \ge 1$. Let $1 \le k < f_b(n)$. Then $(k \times b^n + 1)_b \notin P_b^*$.
\end{lemma}
\begin{proof}
Let $x = (k \times b^n + 1)_b = (k)_b 0^{n-1} 1$. For the sake of contradiction, suppose $x \in P_b^*$. Therefore, there is a suffix $z$ such that $x = yz$ and $z \in P_b$. Consider the cases for the first index of $z$.

If $z$ starts with $0$ then $z \notin P_b$. If $z = 1$ then $z \notin P_b$. Lastly, suppose the first index of $z$ is in $(k)_b$. Then $z$ represents $k' \times b^n + 1$ where $1 \le k' \le k$. Since $1 \le k' \le k < f_b(n)$, $k' \times b^n + 1$ is not prime, so $z \notin P_b$. Contradiction.
\end{proof}

\begin{theorem}\label{PSnotreg}
Let $b \ge 2$. $P_b^*$ is not regular.
\end{theorem}
\begin{proof}
For the sake of contradiction, suppose $P_b^*$ is regular. Let $p \ge 1$ be the pumping length of $P_b^*$. By Proposition \ref{bigf}, let $N \ge 2$ such that $f_b(N) > b^p$. Consider the string:
$$
s
= (f_b(N) \times b^N + 1)_b
= (f_b(N))_b 0^{N-1} 1
$$

Note that $f_b(N) \times b^N + 1$ is prime. Therefore, $s \in P_b$ and $s \in P_b^*$.

Since $f_b(N) > b^p$, we have $|(f_b(N))_b| \ge |(b^p)_b| = p+1$.

By the pumping lemma, there exists a decomposition $s = xyz$ where $|y| \ge 1$, $|xy| \le p$, and $xy^iz \in P_b^*$ for all $i \ge 0$. Consider $xz \in P_b^*$. Since $|xy| < |(f_b(N))_b|$, $y$ is a subword of $(f_b(N))_b$. Therefore, $xz$ represents $k \times b^N + 1$ for some $0 \le k < f_b(N)$. We reject $k = 0$ because $xz$ cannot start with $0$. We reject $1 \le k < f_b(N)$ because Lemma \ref{noP*} implies $xz \notin P_b^*$. Contradiction.
\end{proof}

\begin{corollary}
Let $b \ge 2$. $P_b$ is not regular.
\end{corollary}
\begin{proof}
We can now recover the result by Minsky and Papert. If $L$ is a regular language, then $L^*$ is regular. From Theorem \ref{PSnotreg}, it follows that $P_b$ is not regular.
\end{proof}

\section{Acknowledgment}

I thank Jeffrey Shallit for telling me about this problem.

\end{document}